%% file: main.tex
\documentclass[letterpaper,10pt,conference]{ieeeconf}

\IEEEoverridecommandlockouts
\usepackage[T1]{fontenc}
\usepackage{microtype}

\usepackage{amsmath}
\usepackage{amssymb}
\usepackage{bm}
\usepackage{siunitx}

\usepackage{xcolor}

\usepackage{graphicx}
\usepackage{booktabs}
\usepackage[font=footnotesize]{subfig}
\usepackage{tikz}

\usetikzlibrary{
  arrows.meta,
  shapes,
  calc,
  positioning
}

\usepackage{cite}

\usepackage[hidelinks]{hyperref}

\input{symbols}

\title{
Aerial Non-Stop Trajectories Preserving the Equilibrium of
Loads Suspended by Variable-Length Cables
}
\author{%
	 Antonio Franchi$^{1,2}$, Alberto Landi$^3$, Chiara Gabellieri$^1$
  \thanks{%
  Robotics and Mechatronics (RaM) lab, EEMCS Faculty, University of Twente, 7500 AE Enschede, The Netherlands (e-mail: c.gabellieri@utwente.nl, schol@r-franchi.eu)
  }
  \thanks{%
   Department of
Computer, Control and Management Engineering, Sapienza University of Rome,  00185 Rome, Italy
  }
  \thanks{%
  Department of Information Engineering, University of Pisa, 56126 Pisa, Italy (e-mail: alberto.landi@unipi.it)
  }
  \thanks{%
    This work has been partially funded by the MSCA Staff Exchange project
    NEUTRAWEED under grant No.~101182891.
  }%
}

\newif\ifarXiv
\arXivtrue

\begin{document}

\maketitle

\begin{abstract}
This work studies equilibrium-preserving non-stop trajectories of
aerial carriers connected to rigid loads by variable-length cables.
Internal-force motions vary the cable directions without changing the
load wrench, while cable-length actuation shapes the radial realization
of the carrier trajectories. We derive an explicit cable-length law for
realizing prescribed carrier-position 
profiles along a fixed direction,
with constant altitude as a relevant specialization. Under this
constraint, a carrier stops exactly when its cable direction stops
changing. This characterization reveals a distinction between
unrestricted and level motion: variable cable lengths enable bounded
non-stop motion with two carriers, whereas periodic level non-stop
motion is obstructed for two carriers and constructively achieved for
three or more under explicit geometric and feasibility conditions.
Numerical simulations illustrate the constant-altitude construction.
\end{abstract}

\section{Introduction and Related Work}

Cable-suspended manipulation provides a lightweight and mechanically
simple means of transporting loads of different sizes and shapes
\cite{lv2026multirotor,drones8020035}. Cooperative aerial systems
distribute the payload among multiple vehicles and enable full-pose
manipulation of rigid loads. Two-carrier systems have mainly addressed
slender loads
\cite{pereira2018asymmetric,gabellieri2023equilibria}, whereas
rigid-body manipulation commonly employs three or more carriers
\cite{masone2016,li2021cooperative}.

Most aerial transportation methods rely on hovering 
multirotors, but non-stop motion permits carriers that require a
positive forward speed. Rotating equilibria have been investigated for
energy-efficient transport of point-mass loads
\cite{foss2026energy}, while cable-towed fixed-wing systems directly
motivate persistent motion
\cite{williams2009dynamics,quenneville2023experimental}. Preserving the
pose of a suspended rigid load under this requirement is challenging
because the cable forces must continuously generate the same load
wrench. At least three carriers are required for persistent non-stop motion,
for which constructive coordination methods, geometric
characterizations, and feedback-based extensions are available
\cite{gabellieri2025coordinated,van2026geometry,
girardello2025trajectory} with fixed cable lengths. Level flight is especially
relevant for wing-borne carriers because it avoids repeated,
energy-consuming departures from steady forward flight. 
 With fixed cable lengths, however, the
carriers are constrained to spheres centered at the load-side attachment points,
preventing independent altitude shaping. 
Cable-length
actuation has been studied for payload deployment and maneuvering
\cite{goodarzi2016autonomous,li2023flight,li2020design}, but not for
enforcing level equilibrium-preserving non-stop trajectories.

This work addresses the generation of carrier and variable-cable-length trajectories compatible with a constant load pose. 

The main contributions are: 1) an explicit radial-shaping law for imposing one directional component of each carrier position; 2) a two-carrier construction for bounded unrestricted non-stop motion; 3) an obstruction result for periodic two-carrier constant-altitude motion; and 4) a phase-shifted constant-altitude non-stop construction for three or more carriers under explicit geometric and feasibility conditions. These results concern reference generation and do not establish the dynamic tracking of the resulting trajectories.
Numerical simulations illustrate the
constant-altitude construction.

\section{Modeling}
\label{sec:vl_cables}


{Let \(\Eaff\) be an oriented three-dimensional Euclidean affine space, with associated translation vector space \(\Vtrans\),
endowed with the Euclidean inner product
\(\inner{\cdot}{\cdot}\), induced norm \(\norm{\cdot}\), and 
cross product. The inner product identifies \(\Vtrans\) with
its dual \(\Vtrans^*\). Intrinsic physical \emph{points} belong to
\(\Eaff\), whereas displacements between points, velocities, and
angular velocities are represented by vectors in \(\Vtrans\); forces
and moments are represented by their Euclidean-dual vectors in
\(\Vtrans^*\).}


{A right-handed Cartesian frame
\(\Frame{a}
=\{O_{\mathsf a},x_{\mathsf a},y_{\mathsf a},z_{\mathsf a}\}\)
consists of an origin \(O_{\mathsf a}\in\Eaff\) and an ordered
right-handed orthonormal basis
\((x_{\mathsf a},y_{\mathsf a},z_{\mathsf a})\) of \(\Vtrans\).
It identifies each intrinsic vector \(v\in\Vtrans\) with the coordinate
array
\(\coord{a}{v}
\define
[\inner{x_{\mathsf a}}{v}\ \
 \inner{y_{\mathsf a}}{v}\ \
 \inner{z_{\mathsf a}}{v}]^\top
\in\Rthree\).}
{For a point \(P\in\Eaff\), its \emph{position} in \(\Frame{a}\)
is defined as the coordinate array of the displacement vector
\(\overrightarrow{O_{\mathsf a}P}\) and is denoted by
\(\pointcoord{a}{P}
\define
\coord{a}{\big(\overrightarrow{O_{\mathsf a}P}\big)}
\in\Rthree\).
Thus, \(P\) is an intrinsic point, whereas
\(\pointcoord{a}{P}\) is its coordinate representation relative to
\(\Frame{a}\).} 

{Given two Cartesian frames \(\Frame{a}\) and \(\Frame{c}\), the
rotation matrix \(\rot{a}{c}\in\SOthree\) maps vector coordinates
expressed in \(\Frame{c}\) into coordinates expressed in
\(\Frame{a}\), namely,
\(\coord{a}{v}=\rot{a}{c}\coord{c}{v}\) for every \(v\in\Vtrans\).
Its columns are the coordinate arrays in \(\Frame{a}\) of the basis
vectors of \(\Frame{c}\), that is,
\(\rot{a}{c}
=
[\coord{a}{x_{\mathsf c}}\ \
 \coord{a}{y_{\mathsf c}}\ \
 \coord{a}{z_{\mathsf c}}]\).}
{The rigid transformation from \(\Frame{c}\) to \(\Frame{a}\)
is represented by
\(\pose{a}{c}
\define
\left[
\begin{smallmatrix}
  \rot{a}{c} & \pointcoord{a}{O_{\mathsf c}}\\
  \bm 0_{1\times3} & 1
\end{smallmatrix}
\right]
\in\SEthree\).
Accordingly, the position-coordinate arrays of any point
\(P\in\Eaff\) satisfy
\(
\left[
\begin{smallmatrix}
  \pointcoord{a}{P}\\1
\end{smallmatrix}
\right]
=
\pose{a}{c}
\left[
\begin{smallmatrix}
  \pointcoord{c}{P}\\1
\end{smallmatrix}
\right]
\),
or, equivalently,
\(\pointcoord{a}{P}
=
\pointcoord{a}{O_{\mathsf c}}
+
\rot{a}{c}\pointcoord{c}{P}\).}


Let
{\(\frameW
=\{\originW,x_{\mathsf w},y_{\mathsf w},z_{\mathsf w}\}\)}
be an inertial frame with origin \(\originW\), and let
{\(\frameB
=\{\originB,x_{\mathsf b},y_{\mathsf b},z_{\mathsf b}\}\)}
be a body-fixed frame whose origin \(\originB\) coincides with the
{center} of mass (CoM) of the suspended rigid body.
{The load pose is represented by its position and orientation
components \(\loadPos\in\Rthree\) and \(\loadRot\in\SOthree\),
respectively, or, equivalently, by
\(
  \loadPose
  =
  \left[
  \begin{smallmatrix}
    \loadRot & \loadPos\\
    \bm 0_{1\times3} & 1
  \end{smallmatrix}
  \right]
  \in\SEthree
\).}
The angular velocity of the load
{relative to} the inertial frame, expressed in \(\frameB\), is
denoted by {\(\loadAngVel\)}.
The mass and
{the body-frame coordinate matrix of the}
rotational inertia tensor of the load are denoted by
{\(\loadMass\) and \(\loadInertia\)}, respectively.

The rigid body is manipulated through \(n\) massless cables, each of
variable length
{\(\cableLength{i}(t)>0\)}
(with {\(i\in\{1,\ldots,n\}\)}) and attached to
{points \(\anchorPoint{i}\in\Eaff\)}
on the object, where
{%
\(
  \attachmentCoord{i}
  \define
  \coord{b}{\big(\overrightarrow{\originB\anchorPoint{i}}\big)}
  \in\Rthree
\)
denotes the constant body-frame coordinate array of the displacement
from \(\originB\) to \(\anchorPoint{i}\).} {Let \(\carrierPoint{i}\in\Eaff\) denote the carrier-side
endpoint of the \(i\)-th cable, and let
\(\cableDirIntrinsic{i}\in\Vtrans\) be the unit vector directed from
\(\anchorPoint{i}\) to \(\carrierPoint{i}\).
If the cable is taut and straight, then
\(
  \overrightarrow{\anchorPoint{i}\carrierPoint{i}}
  =
  \cableLength{i}\cableDirIntrinsic{i}
\)
and
\(\norm{\cableDirIntrinsic{i}}=1\).
The world-frame coordinate array of \(\cableDirIntrinsic{i}\) is
denoted by \(\cableDir{i}\in\sphere{2}\).}
{Finally, let \(\cableTension{i}>0\) denote the corresponding
cable tension. The force exerted by cable
\(i\) on the load is}
\begin{equation}
{%
  \cableForce{i}
  =
  \cableTension{i}\cableDir{i}.
}
\label{eq:fi1}
\end{equation}


The world-frame position of the aerial carrier connected to the free
end of the \(\ith{i}\) cable is
\begin{equation}
{%
  \carrierPos{i}
  =
  \loadPos
  +
  \loadRot\attachmentCoord{i}
  +
  \cableLength{i}\cableDir{i}.
}
\label{eq:carrier_position}
\end{equation}
{Equation~\eqref{eq:carrier_position} is the world-frame
coordinate representation of the intrinsic affine relation
\(
  \overrightarrow{\originW\carrierPoint{i}}
  =
  \overrightarrow{\originW\originB}
  +
  \overrightarrow{\originB\anchorPoint{i}}
  +
  \overrightarrow{\anchorPoint{i}\carrierPoint{i}}
\).}

From~\eqref{eq:carrier_position}, the cable direction is obtained as
\begin{equation}
{%
  \cableDir{i}
  =
  \tfrac{1}{\cableLength{i}}
  \big(
    \carrierPos{i}
    -
    \loadPos
    -
    \loadRot\attachmentCoord{i}
  \big).
}
\label{eq:qiW}
\end{equation}
When the load pose is constant,
{%
\(
  \dloadPos=\bm 0
\)
and
\(
  \frac{\mathrm d}{\mathrm dt}\loadRot=\bm 0
\),
}
\begin{equation}
{%
  \dcarrierPos{i}
  =
  \cableLength{i}\dcableDir{i}
  +
  \dcableLength{i}\cableDir{i}.
}
\label{eq:carrier_velocity_static}
\end{equation}
Equation~\eqref{eq:carrier_velocity_static} shows that the carrier
velocity consists of two components: the term
{\(\cableLength{i}\dcableDir{i}\)} is tangential to the sphere
centered at the cable attachment point and is generated by variations
of the cable direction, whereas the term
{\(\dcableLength{i}\cableDir{i}\)} is radial and is generated
directly by winch actuation.
{Since
\(\left(\cableDir{i}\right)^\top\dcableDir{i}=0\), these two
components are orthogonal and}
\begin{equation}
{%
  \norm{\dcarrierPos{i}}^2
  =
  \cableLength{i}^{\,2}\norm{\dcableDir{i}}^2
  +
  \dcableLength{i}^{\,2}.
}
\label{eq:carrier_speed_decomposition}
\end{equation}
Therefore, with variable-length cables, the carrier can remain in
motion even when the tangential component vanishes, provided that the
cable length is changing. In the fixed-length case,
{\(\dcableLength{i}=0\)}, Eq.~\eqref{eq:carrier_velocity_static}
reduces to
{%
\(
  \dcarrierPos{i}
  =
  \cableLength{i}\dcableDir{i}
\),
}
which is the {simplified} relation considered in
\cite{gabellieri2025coordinated}.

Throughout the paper, a carrier trajectory is called \emph{non-stop}
if there exists a constant
{\(\speedLowerBound>0\)}
such that the speed of every carrier remains uniformly bounded away
from zero, i.e.,
\[
  \norm{\dcarrierPos{i}(t)}
  \geq
  \speedLowerBound,
  \qquad
  \forall t\geq0,
  \quad
  \forall i\in\{1,\ldots,n\}.
\]
This condition excludes both isolated stops and trajectories whose
speed asymptotically approaches zero.

Differentiating~\eqref{eq:fi1} gives the kinematic-force relation
\begin{equation}
{%
  \dcableForce{i}
  =
  \dcableTension{i}\cableDir{i}
  +
  \cableTension{i}\dcableDir{i}
  \define
  \big(\dcableForce{i}\big)^\parallel
  +
  \big(\dcableForce{i}\big)^\perp ,
}
\label{eq:fidot1_base}
\end{equation}
where
{%
\(
  \big(\dcableForce{i}\big)^\parallel
  =
  \dcableTension{i}\cableDir{i}
\)
}
is parallel to the cable direction and represents a variation of the
force magnitude caused by a change in cable tension, whereas
{%
\(
  \big(\dcableForce{i}\big)^\perp
  =
  \cableTension{i}\dcableDir{i}
\)
}
is orthogonal to the cable direction and represents a variation of the
force direction caused by the time variation
{\(\dcableDir{i}\)}
of the cable-direction coordinate array.
{The following sections exploit these complementary components:
equilibrium-preserving force variations generate the cable-direction
motion, whereas cable-length actuation shapes its radial realization.}


{We choose the positive \(z_{\mathsf w}\)-axis opposite to the
gravity direction and assume a uniform gravitational acceleration of
magnitude \(\gravityScalar>0\).}
The load translational dynamics are 
\begin{equation}
{%
  \loadMass\ddloadPos
  =
  -\loadMass\gravityScalar\worldVertical
  +
  \sum_{i=1}^{n}\cableForce{i},
}
\label{eq:load_trans}
\end{equation}
whereas the rotational dynamics
{about the load CoM, with all vector coordinates}
expressed in the body frame, are
\begin{equation}
{%
  \loadInertia\dloadAngVel
  =
  -\loadAngVel
  \times
  \left(
    \loadInertia\loadAngVel
  \right)
  +
  \sum_{i=1}^{n}
  \left[
    \attachmentCoord{i}
  \right]_{\times}
  \rot{b}{w}\cableForce{i}.
}
\label{eq:load_rot}
\end{equation}
Here,
{\([\cdot]_{\times}\)}
is the skew-symmetric operator such that
{%
\(
  [\bm x]_{\times}\bm y
  =
  \bm x\times\bm y
\).
}
Thus,
{%
\(
  [\attachmentCoord{i}]_{\times}
  \rot{b}{w}\cableForce{i}
\)
}
is the body-frame coordinate array of the \(i\)-th cable moment about
the load CoM.

\section{Equilibrium-Preserving Force Motions}
\label{sec:equilibrium_force_motions}

Let us consider a feasible static equilibrium, at which the wrench
applied to the load by the cables is constant.

\begin{definition}[Internal-force space]
\label{def:internal_force_space}
{Fixing the load CoM \(\originB\) as the moment pole, the
intrinsic \emph{grasp map} is
\(
  \Grasp:\Vtrans^n\rightarrow\WrenchSpace
\),
with
\(\WrenchSpace\define\Vtrans\times\Vtrans\), defined by
\[
  \Grasp(f_1,\ldots,f_n)
  \define
  \left(
    \sum_{i=1}^{n}f_i,\,
    \sum_{i=1}^{n}
    \overrightarrow{\originB\anchorPoint{i}}\times f_i
  \right).
\]
Its first and second components are the resultant force and the
resultant moment about \(\originB\), respectively. The intrinsic
internal-force space is
\(
  \ker\Grasp\subset\Vtrans^n
\).
}
\end{definition}

The world-frame cable-force arrays are stacked as
\begin{equation}
{%
  \cableForceStack
  \define
  \begin{bmatrix}
    {\cableForce{1}}^\top
    &
    \cdots
    &
    {\cableForce{n}}^\top
  \end{bmatrix}^{\top}.
}
\label{eq:stacked_cable_forces}
\end{equation}
The equilibrium condition reads
\begin{equation}
{%
  \grasp\cableForceStack
  =
  \equilibriumWrench,
}
\label{eq:wrench_equilibrium}
\end{equation}
where
{\(\grasp\in\mathbb{R}^{6\times3n}\) is the world-frame matrix
of \(\Grasp\):}
\begin{equation}
{%
  \grasp
  \define
  \begin{bmatrix}
    \bm I_3
    &
    \cdots
    &
    \bm I_3
    \\
    [\attachmentWorld{1}]_{\times}
    &
    \cdots
    &
    [\attachmentWorld{n}]_{\times}
  \end{bmatrix},
  \qquad
  \attachmentWorld{i}
  \define
  \loadRot\attachmentCoord{i}.
}
\label{eq:grasp_matrix}
\end{equation}
{Since the load pose is constant, \(\grasp\) is constant.}
Moreover,
\begin{equation}
{%
  \equilibriumWrench
  \define
  \begin{bmatrix}
    \loadMass\gravityScalar\worldVertical\\
    \bm 0_{3\times1}
  \end{bmatrix}
}
\label{eq:equilibrium_wrench}
\end{equation}
is the
{world-frame coordinate array of the cable wrench required to
balance gravity about \(\originB\)}.

{Define the admissible tensile-force set as}
\begin{equation}
{%
  \mathcal F_+
  \define
  \left\{
    \col{\bm f_1^{\mathsf w},\ldots,\bm f_n^{\mathsf w}}
    \ \middle|\
    \norm{\bm f_i^{\mathsf w}}>0,\ i=1,\ldots,n
  \right\}.
}
\label{eq:admissible_tensile_force_set}
\end{equation}
{The condition
\(\norm{\bm f_i^{\mathsf w}}>0\) is equivalent to
\(\cableTension{i}>0\), because
\(\cableForce{i}=\cableTension{i}\cableDir{i}\) and
\(\norm{\cableDir{i}}=1\).}
The prescribed static equilibrium is assumed feasible:
${%
  \equilibriumWrench
  \in
  \grasp\mathcal F_+
  \subseteq
  \im\grasp.
}$

\begin{prop}[Tensile equilibrium-force family]
\label{prop:equilibrium_force_family}
Let
{\(\particularForce\in\mathcal F_+\)}
and {\(\nullBasis\)} satisfy
\begin{equation}
{%
  \grasp\particularForce
  =
  \equilibriumWrench,
  \qquad
  \im\nullBasis
  =
  \ker\grasp.
}
\label{eq:particular_force_and_null_basis}
\end{equation}
Then, every tensile force distribution producing the prescribed
equilibrium wrench can be written as
\begin{equation}
{%
  \cableForceStack
  =
  \particularForce
  +
  \nullBasis\internalParam
  \in
  \mathcal F_+,
  \qquad
  \internalParam
  \in
  \mathbb{R}^{\dim\ker\grasp}.
}
\label{eq:forces_general}
\end{equation}
\end{prop}

\begin{proof}
{Equation~\eqref{eq:particular_force_and_null_basis} gives one
solution of \eqref{eq:wrench_equilibrium}, while
\(\im\nullBasis=\ker\grasp\) contains exactly the force variations
that generate zero resultant wrench. Intersecting the resulting
affine solution space with \(\mathcal F_+\) enforces tensile
admissibility.}
\end{proof}

\begin{definition}[Internal force]
\label{def:internal_force}
The internal-force array and its \(i\)-th block are defined by
\begin{equation}
{%
  \internalForceStack
  \define
  \nullBasis\internalParam
  \in
  \ker\grasp,
  \qquad
  \internalForceBlock{i}
  \define
  \big(\internalForceStack\big)_i.
}
\label{eq:internal_force_definition}
\end{equation}
Hence,
{%
\(
  \cableForceStack
  =
  \particularForce+\internalForceStack
\)
and
\(
  \grasp\internalForceStack=\bm0_{6\times1}
\).
}
Internal forces can therefore change the individual cable forces and
directions while preserving the load wrench.
\end{definition}

\subsection{Cycle-Induced Internal-Force Motions}

Let
\(\mathcal K_n=(\mathcal V,\mathcal V^2)\) be the complete directed
graph with vertex set
\(\mathcal V\define\{1,\ldots,n\}\), where vertex
\(i\in\mathcal V\) corresponds to carrier \(i\).
Consider a directed Hamiltonian cycle \(\hamCycle\), and let
{\(\incidence\in\{-1,0,1\}^{n\times n}\)}
be its incidence matrix.
Denote the cycle's ordered edges by
{\(e_1,\ldots,e_n\)}.

For a directed edge
{\(e_j=(e_j^1,e_j^2)\)}, let
{\(\edgeDir{e_j^1}{e_j^2}\)} be the world-frame coordinate array
of the unit vector  from \(\anchorPoint{e_j^2}\) to
\(\anchorPoint{e_j^1}\):
\begin{equation}
{%
  \edgeDir{e_j^1}{e_j^2}
  \define
  \norm{
    \attachmentCoord{e_j^1}
    -
    \attachmentCoord{e_j^2}
  }^{-1}
  \loadRot
  \big(
    \attachmentCoord{e_j^1}
    -
    \attachmentCoord{e_j^2}
  \big).
}
\label{eq:edge_unit_direction}
\end{equation}

\begin{definition}[Cycle-induced force matrix]
\label{def:cycle_force_matrix}
{The cycle-induced internal-force matrix is}
\begin{equation}
{%
  \cycleNullMatrix(\hamCycle)
  \define
  (\incidence\otimes\bm I_3)
  \operatorname{diag}
  \left(
    \edgeDir{e_1^1}{e_1^2},
    \ldots,
    \edgeDir{e_n^1}{e_n^2}
  \right)
  \in\mathbb{R}^{3n\times n}.
}
\label{eq:choice_of_N}
\end{equation}
\end{definition}

\begin{lemma}[Edge self-equilibration]
\label{lemma:edge_self_equilibration}
{Every column of \(\cycleNullMatrix(\hamCycle)\) has zero
resultant force and zero resultant moment about \(\originB\).}
\end{lemma}

\begin{proof}
By \eqref{eq:choice_of_N}, each column applies two equal and opposite
forces, so its resultant force in \eqref{eq:stacked_cable_forces} is zero.
If the corresponding edge connects \(\anchorPoint{k}\) and
\(\anchorPoint{l}\), its resultant moment is, up to the common force
magnitude,
$
  \big(
    \attachmentWorld{k}
    -
    \attachmentWorld{l}
  \big)
  \times
  \edgeDir{k}{l}
  =
  \bm 0,
$
because the two factors are parallel by
\eqref{eq:edge_unit_direction} and~\eqref{eq:grasp_matrix}.
\end{proof}

\begin{prop}[Cycle-induced internal-force subspace]
\label{prop:cycle_internal_force_subspace}
The cycle-induced internal-force matrix satisfies
\begin{equation}
{%
  \im\cycleNullMatrix(\hamCycle)
  \subseteq
  \ker\grasp.
}
\label{eq:cycle_subspace_in_kernel}
\end{equation}
Consequently,
\begin{equation}
{%
  \cableForceStack
  =
  \particularForce
  +
  \cycleNullMatrix(\hamCycle)\cycleParam
  \in
  \mathcal F_+
}
\label{eq:cycle_force_distribution}
\end{equation}
preserves the prescribed equilibrium wrench for every
{\(\cycleParam\in\mathbb{R}^n\)}
for which the resulting force distribution remains tensile.
\end{prop}

\begin{proof}
{Lemma~\ref{lemma:edge_self_equilibration} shows that every
column of \(\cycleNullMatrix(\hamCycle)\) belongs to
\(\ker\grasp\), which proves
\eqref{eq:cycle_subspace_in_kernel}. Equation
\eqref{eq:cycle_force_distribution} then follows from
\(\grasp\particularForce=\equilibriumWrench\).}
\end{proof}

Thus, \(\cycleNullMatrix(\hamCycle)\) parametrizes a cycle-induced internal-force subspace. The subsequent constructions use this subspace without requiring a full basis of the grasp kernel.

\begin{corollary}[Local two-direction representation]
\label{cor:local_internal_force_representation}
For the Hamiltonian cycle \(\hamCycle\), let \(h_i\) and \(h_i^+\)
denote the indices of the incoming and outgoing edges at vertex \(i\).
The corresponding signed versions of the direction arrays defined in
\eqref{eq:edge_unit_direction} are
\begin{equation}
  \coord{w}{\delta_i}
  \define
  \bm N_{H,i h_i}^{\mathsf w},
  \qquad
  \coord{w}{\bar{\delta}_i}
  \define
  \bm N_{H,i h_i^+}^{\mathsf w}.
\label{eq:local_cycle_directions}
\end{equation}
Then, the internal-force block associated with carrier \(i\) is
\begin{equation}
  \internalForceBlock{i}
  =
  \lambda_{h_i}\coord{w}{\delta_i}
  +
  \lambda_{h_i^+}\coord{w}{\bar{\delta}_i}.
\label{eq:time_var_force_Ni}
\end{equation}
\end{corollary}

\begin{proof}
Every vertex of \(\hamCycle\) has one incoming and one outgoing edge.
Hence, by \eqref{eq:choice_of_N}, the \(i\)-th block row of
\(\cycleNullMatrix(\hamCycle)\) has exactly the two nonzero blocks in
\eqref{eq:local_cycle_directions}.
\end{proof}

\section{Radial Shaping of Carrier Trajectories}
\label{sec:radial_shaping}

{Consider a continuously differentiable tensile
equilibrium-preserving force trajectory
\(\cableForceStack(t)\in\mathcal F_+\), generated, for example, by
\eqref{eq:cycle_force_distribution}. It determines the cable-direction
trajectories through
\(
  \cableDir{i}
  =
  \cableForce{i}/\norm{\cableForce{i}}
\).
The cable lengths can then be selected independently to shape the
radial realization of these directions, subject to their admissible
bounds.}

\begin{prop}[Directional-position realization]
\label{prop:directional_position_realization}

{Let \(u\in\Vtrans\) be a fixed unit vector, and let
\(\coord{w}{u}\in\sphere{2}\) be its world-frame coordinate array.
Define}
\begin{equation}
{%
  s_i(t)
  \define
  \left(\coord{w}{u}\right)^\top
  \carrierPos{i}(t),
  \qquad
  q_{i,u}(t)
  \define
  \left(\coord{w}{u}\right)^\top
  \cableDir{i}(t).
}
\label{eq:directional_coordinates}
\end{equation}
{Let \(s_i^\star(t)\) be a continuously differentiable desired
profile satisfying
\(
  \abs{q_{i,u}(t)}\geq\varepsilon_i>0
\)
for all \(t\geq0\). Then \(s_i(t)=s_i^\star(t)\) is realized by}
\begin{equation}
{%
  \cableLength{i}(t)
  =
  \frac{
    s_i^\star(t)-s_{B_i}
  }{
    q_{i,u}(t)
  },
  \qquad
  s_{B_i}
  \define
  \left(\coord{w}{u}\right)^\top
  \left(
    \loadPos+\attachmentWorld{i}
  \right),
}
\label{eq:directional_length_realization}
\end{equation}
{provided that}
\begin{equation}
{%
  0<\cableLengthMin
  \leq
  \cableLength{i}(t)
  \leq
  \cableLengthMax
  <\infty.
}
\label{eq:directional_length_bounds}
\end{equation}
\end{prop}

\begin{proof}
{Since the load pose is constant,
\(\loadPos+\attachmentWorld{i}\) is the world-frame position of
\(\anchorPoint{i}\). Projecting \eqref{eq:carrier_position} along
\(\coord{w}{u}\) yields}
\begin{equation}
{%
  s_i(t)
  =
  s_{B_i}
  +
  \cableLength{i}(t)q_{i,u}(t).
}
\label{eq:directional_position_identity}
\end{equation}
{Solving
\(s_i(t)=s_i^\star(t)\) for \(\cableLength{i}(t)\) gives
\eqref{eq:directional_length_realization}. The lower bound on
\(\abs{q_{i,u}}\) prevents singular length commands, while
\eqref{eq:directional_length_bounds} ensures their physical
admissibility.}
\end{proof}

\begin{corollary}[Fixed-length compatibility]
\label{cor:fixed_length_directional_compatibility}

{If cable \(i\) has constant length
\(\cableLength{i}(t)\equiv\ell_i\), then a prescribed directional
profile is feasible only if}
\begin{equation}
{%
  s_i^\star(t)-s_{B_i}
  =
  \ell_i q_{i,u}(t).
}
\label{eq:fixed_length_directional_compatibility}
\end{equation}
{Thus, fixed cable length constrains the directional profile to
that induced by the cable-direction trajectory, whereas variable
cable length permits its radial shaping through
\eqref{eq:directional_length_realization}.}
\end{corollary}

\begin{proof}
{The claim follows from
\eqref{eq:directional_position_identity} by imposing
\(\cableLength{i}(t)\equiv\ell_i\).}
\end{proof}

\subsection{Constant-Altitude Realization}

{A practically relevant specialization is obtained by selecting
\(u=z_{\mathsf w}\). Define the altitude of the \(i\)-th carrier and
its load-side attachment point as}
\begin{equation}
{%
  z_{A_i}(t)
  \define
  {\worldVertical}^\top\carrierPos{i}(t),
  \qquad
  z_{B_i}
  \define
  {\worldVertical}^\top
  \left(
    \loadPos+\attachmentWorld{i}
  \right),
}
\label{eq:carrier_anchor_altitudes}
\end{equation}
{and let}
$
{%
  q_{i,z}(t)
  \define
  {\worldVertical}^\top\cableDir{i}(t)
}$
{denote the vertical component of the cable direction.}

\begin{corollary}[Constant-altitude realization]
\label{cor:constant_altitude_realization}

{Let \(h_i\) be a desired constant carrier altitude. If}
\begin{equation}
{%
  \abs{q_{i,z}(t)}
  \geq
  \varepsilon_i
  >0,
  \qquad
  \forall t\geq0,
}
\label{eq:vertical_direction_nondegeneracy}
\end{equation}
{then \(z_{A_i}(t)=h_i\) is realized by}
\begin{equation}
{%
  \cableLength{i}(t)
  =
  \frac{
    h_i-z_{B_i}
  }{
    q_{i,z}(t)
  },
}
\label{eq:constant_altitude_length}
\end{equation}
{provided that the resulting length satisfies
\eqref{eq:directional_length_bounds}.}
\end{corollary}

\begin{proof}
{Apply Proposition~\ref{prop:directional_position_realization}
with \(u=z_{\mathsf w}\) and
\(s_i^\star(t)\equiv h_i\).}
\end{proof}

{In the fixed-length case,
\eqref{eq:fixed_length_directional_compatibility} reduces to
\(
  h_i-z_{B_i}
  =
  \ell_iq_{i,z}(t)
\).
Therefore, constant-altitude motion with fixed cable length requires
\(q_{i,z}(t)\) to be constant. Geometrically, the cable direction is
then restricted to a constant-latitude circle on \(\sphere{2}\).
Variable cable length removes this restriction by keeping the product
\(\cableLength{i}(t)q_{i,z}(t)\) constant while both factors may vary.}

\begin{prop}[Stopping condition under constant altitude]
\label{prop:constant_altitude_stopping_condition}

{Suppose that the constant-altitude length law
\eqref{eq:constant_altitude_length} satisfies
\eqref{eq:vertical_direction_nondegeneracy} and the admissible length
bounds. Then}
\begin{equation}
{%
  \dcarrierPos{i}(t)=\bm 0
  \quad\Longleftrightarrow\quad
  \dcableDir{i}(t)=\bm 0.
}
\label{eq:constant_altitude_stopping_equivalence}
\end{equation}
\end{prop}

\begin{proof}
{Constant altitude implies}
${%
  \cableLength{i}(t)q_{i,z}(t)
  =
  h_i-z_{B_i}.
}$
Differentiating it and using
\(q_{i,z}(t)\neq0\) gives
\begin{equation}
{%
  \dcableLength{i}
  =
  -
  \cableLength{i}
  \frac{
    \dot q_{i,z}
  }{
    q_{i,z}
  }.
}
\label{eq:constant_altitude_length_rate}
\end{equation}
{Substitution into \eqref{eq:carrier_velocity_static} yields}
\begin{equation}
{%
  \dcarrierPos{i}
  =
  \cableLength{i}
  \left(
    \dcableDir{i}
    -
    \frac{
      \dot q_{i,z}
    }{
      q_{i,z}
    }
    \cableDir{i}
  \right).
}
\label{eq:constant_altitude_velocity}
\end{equation}
{If \(\dcableDir{i}=\bm0\), then
\(\dot q_{i,z}=0\), and
\eqref{eq:constant_altitude_velocity} gives
\(\dcarrierPos{i}=\bm0\).
Conversely, if \(\dcarrierPos{i}=\bm0\), then
\eqref{eq:constant_altitude_velocity} implies that
\(\dcableDir{i}\) is parallel to \(\cableDir{i}\).
Since
\(
  {\cableDir{i}}^{\top}\dcableDir{i}=0
\),
this is possible only if
\(\dcableDir{i}=\bm0\).}
\end{proof}

\begin{corollary}[Level non-stop condition]
\label{cor:level_nonstop_condition}

{Let the cable-direction and induced cable-length trajectories
be continuous and periodic and satisfy the assumptions of
Proposition~\ref{prop:constant_altitude_stopping_condition}. A
constant-altitude carrier trajectory is non-stop if and only if}
\begin{equation}
{%
  \dcableDir{i}(t)\neq\bm0,
  \qquad
  \forall t\geq0.
}
\label{eq:level_nonstop_direction_condition}
\end{equation}
\end{corollary}

\begin{proof}
{By Proposition~\ref{prop:constant_altitude_stopping_condition},
the carrier speed vanishes exactly when the cable-direction rate
vanishes. Hence, \eqref{eq:level_nonstop_direction_condition} implies
that the continuous and periodic carrier speed is positive over one
period. It therefore attains a strictly positive minimum, which is the
required uniform speed lower bound. The converse follows directly
from \eqref{eq:constant_altitude_stopping_equivalence}.}
\end{proof}

\section{Structural Existence Results}
\label{sec:structural_existence}

\begin{definition}[Admissible non-stop motion]
\label{def:admissible_nonstop_motion}
{An equilibrium-preserving carrier motion is called
\emph{admissible and non-stop} if the cable forces remain in
\(\mathcal F_+\), the cable lengths satisfy
$  0<\cableLengthMin
  \leq
  \cableLength{i}(t)
  \leq
  \cableLengthMax
  <\infty$ for $
  i=1,\ldots,n,
$
and there exists a common constant \(\speedLowerBound>0\) such that
\[
  \norm{\dcarrierPos{i}(t)}
  \geq
  \speedLowerBound,
  \qquad
  \forall t\geq0,
  \quad
  i=1,\ldots,n.
\]
It is called \emph{level} if every carrier additionally satisfies
\(z_{A_i}(t)=h_i\) for some constant altitude \(h_i\).}
\end{definition}

\subsection{Unrestricted Non-stop Motion}

\begin{prop}[One-carrier impossibility]
\label{prop:one_carrier_impossibility}
{An admissible non-stop motion does not exist for \(n=1\).}
\end{prop}

\begin{proof}
With a single cable and a constant load pose, the translational
equilibrium uniquely determines the cable force. Hence, the cable
direction is constant, i.e.,
{\(\dcableDir{1}=\bm0\)}.
Equation~\eqref{eq:carrier_velocity_static} therefore gives $%
  \dcarrierPos{1}
  =
  \dcableLength{1}\cableDir{1},$ $
  \norm{\dcarrierPos{1}}
  =
  \abs{\dcableLength{1}}.
$

If the motion were non-stop, then
\(\abs{\dcableLength{1}(t)}\geq\speedLowerBound>0\) for every
\(t\geq0\). Since \(\dcableLength{1}\) is continuous and never
vanishes, it has constant sign. Thus, \(\cableLength{1}(t)\) either
increases or decreases at a rate whose magnitude is bounded away from
zero, eventually violating its upper or positive lower bound.
\end{proof}

{Unlike the constant-length case, cable-length actuation can
temporarily move a single carrier while keeping the load still.
Proposition~\ref{prop:one_carrier_impossibility} shows that such motion
cannot remain both non-stop and bounded for all time.}

\subsubsection{Two carriers}

For \(n=2\), assume
{\(\anchorPoint{1}\neq\anchorPoint{2}\)}
and define the world-frame unit direction
\begin{equation}
{%
  \edgeDir{1}{2}
  \define
  \norm{\attachmentCoord{1}-\attachmentCoord{2}}^{-1}
  \loadRot
  \big(
    \attachmentCoord{1}-\attachmentCoord{2}
  \big).
}
\label{eq:internal_direction_n2}
\end{equation}
The one-dimensional internal-force space admits the representation
\begin{equation}
{%
  \cableForce{1}
  =
  \particularForceBlock{1}
  +
  \lambda\edgeDir{1}{2},
  \qquad
  \cableForce{2}
  =
  \particularForceBlock{2}
  -
  \lambda\edgeDir{1}{2}.
}
\label{eq:forces_n2}
\end{equation}
We impose the nondegeneracy condition
\begin{equation}
{%
  \particularForceBlock{i}
  \notin
  \operatorname{span}\{\edgeDir{1}{2}\},
  \qquad
  i=1,2.
}
\label{eq:assumption_n2}
\end{equation}

\begin{lemma}[Two-carrier direction motion]
\label{lemma:two_carrier_direction_motion}
Under \eqref{eq:assumption_n2},
\begin{equation}
{%
  \dcableDir{i}(t)=\bm0
  \quad\Longleftrightarrow\quad
  \dot\lambda(t)=0,
  \qquad
  i=1,2.
}
\label{eq:qdot_lambda_equivalence_n2}
\end{equation}
\end{lemma}

\begin{proof}
Since the cable is tensile,
\(
  \cableDir{i}
  =
  \cableForce{i}/\norm{\cableForce{i}}
\).
Differentiating \eqref{eq:forces_n2} with respect to \(\lambda\)
gives
\begin{equation}
  \frac{\partial\cableDir{i}}{\partial\lambda}
  =
\norm{\cableForce{i}}^{-1}
    \big(
      \bm I_3-\cableDir{i}{\cableDir{i}}^\top
    \big)
    (\pm\edgeDir{1}{2}).
\label{eq:dq_dlambda_n2}
\end{equation}
By \eqref{eq:forces_n2}, the component of \(\cableForce{i}\)
orthogonal to \(\edgeDir{1}{2}\) equals that of
\(\particularForceBlock{i}\), which is nonzero by
\eqref{eq:assumption_n2}. Thus, \(\cableDir{i}\) is never parallel to
\(\edgeDir{1}{2}\), which is therefore not annihilated by the
projector in \eqref{eq:dq_dlambda_n2}. Hence,
\(
  \partial\cableDir{i}/\partial\lambda\neq\bm0
\)
for every admissible \(\lambda\). Since
\(
  \dcableDir{i}
  =
  (\partial\cableDir{i}/\partial\lambda)\dot{\lambda}
\),
\eqref{eq:qdot_lambda_equivalence_n2} follows.
\end{proof}

\begin{prop}[Unrestricted non-stop motion with two carriers]
\label{prop:two_carrier_unrestricted_motion}
Under \eqref{eq:assumption_n2}, an admissible non-stop motion exists
for \(n=2\).
\end{prop}

\begin{proof}
Choose
\begin{equation}
{%
  \lambda(t)
  =
  \lambda_0+a\cos(\xi t),
  \qquad
  \cableLength{i}(t)
  =
  \ell_{i0}+b_i\sin(\xi t),
}
\label{eq:two_carrier_unrestricted_trajectory}
\end{equation}
where \(a\neq0\), \(b_i\neq0\), and \(\xi>0\).
The constants are selected so that the resulting force distribution
remains in \(\mathcal F_+\) and the cable lengths remain strictly
within their admissible bounds.

The derivatives
$
  \dot\lambda(t)
  =
  -a\xi\sin(\xi t),
  \dcableLength{i}(t)
  =
  b_i\xi\cos(\xi t)
$
never vanish simultaneously. Therefore, by
Lemma~\ref{lemma:two_carrier_direction_motion}, neither do
\(\dcableDir{i}\) and \(\dcableLength{i}\). Equation~\eqref{eq:carrier_speed_decomposition} gives
\[
  \norm{\dcarrierPos{i}(t)}^2
  =
  \cableLength{i}(t)^2\norm{\dcableDir{i}(t)}^2
  +
  \dcableLength{i}(t)^2
  >0.
\]
The resulting speeds are continuous and periodic. Their smallest
minimum over one period and over the two carriers therefore provides
a common positive speed lower bound.
\end{proof}

\begin{corollary}[Minimum number for unrestricted motion]
\label{cor:minimum_unrestricted_carriers}
Under the tensile-feasibility and nondegeneracy
conditions, an admissible non-stop motion exists if and only if
\(n\geq2\).
\end{corollary}

\begin{proof}
Necessity follows from
Proposition~\ref{prop:one_carrier_impossibility}.
Proposition~\ref{prop:two_carrier_unrestricted_motion} proves
sufficiency for \(n=2\). For every \(n\geq3\), existence follows from
the fixed-length construction in
\cite{gabellieri2025coordinated}, admissible here by
choosing constant cable lengths.
\end{proof}

\subsection{Structural Limitation of Two-Carrier Level Motion}

\begin{prop}[Obstruction to two-carrier periodic level flight]
\label{prop:two_carrier_level_obstruction}
{Suppose that \(n=2\), \eqref{eq:assumption_n2} holds, and the
bounded internal-force parameter \(\lambda(t)\) is continuously
differentiable and remains in a compact admissible interval. Then no
periodic level carrier motion can be non-stop for all time.}
\end{prop}

\begin{proof}
For level motion, Proposition~\ref{prop:constant_altitude_stopping_condition}
and Lemma~\ref{lemma:two_carrier_direction_motion} give
$
  \dcarrierPos{i}(t)=\bm0
  \quad\Longleftrightarrow\quad
  \dot\lambda(t)=0.
$
Every continuously differentiable periodic scalar function has a
stationary point in each period. Hence, \(\dot\lambda(t)=0\) at some
instant, at which both carriers stop.
\end{proof}
{The obstruction in
Proposition~\ref{prop:two_carrier_level_obstruction} does not contradict
Proposition~\ref{prop:two_carrier_unrestricted_motion}. In the
unrestricted construction, \(\dcableLength{i}\) is deliberately
nonzero whenever \(\dcableDir{i}\) vanishes. Under exact constant
altitude, the cable-length rate is instead fixed by
\eqref{eq:constant_altitude_length_rate} and cannot independently
compensate for a stationary cable direction.}

\subsection{Level Non-stop Motion with \(n\geq3\)}

For \(n\geq3\), define the local internal-force plane
\begin{equation}
  \internalPlane{i}
  \define
  \operatorname{span}
  \left\{
    \coord{w}{\delta_i},
    \coord{w}{\bar{\delta}_i}
  \right\}
  \subset\Rthree.
\label{eq:internal_force_plane}
\end{equation}

\begin{lemma}[Force motion induces direction motion]
\label{lemma:force_direction_motion}
Consider the cycle-based force family
\eqref{eq:cycle_force_distribution} for \(n\geq3\). If
\(\particularForceBlock{i}\notin\internalPlane{i}\), then
$  \dcableForce{i}\neq\bm0
  \quad\Longrightarrow\quad
  \dcableDir{i}\neq\bm0.
$
\end{lemma}

\begin{proof}
From Corollary~\ref{cor:local_internal_force_representation},
\begin{equation}
  \dcableForce{i}
  =
  \dot\lambda_{h_i}\coord{w}{\delta_i}
  +
  \dot\lambda_{h_i^+}\coord{w}{\bar{\delta}_i}
  \in
  \internalPlane{i}.
\label{eq:fdot_plane}
\end{equation}
Moreover,
\(
  \cableForce{i}
  =
  \particularForceBlock{i}
  +
  \internalForceBlock{i}
\),
where
\(\internalForceBlock{i}\in\internalPlane{i}\).
Thus, \(\cableForce{i}\) retains a nonzero component orthogonal to
\(\internalPlane{i}\), so any nonzero
\(\dcableForce{i}\in\internalPlane{i}\) cannot be parallel to
\(\cableForce{i}\). Since
\begin{equation}
  \dcableDir{i}
  =
  \norm{\cableForce{i}}^{-1}
  \big(
    \bm I_3-\cableDir{i}{\cableDir{i}}^\top
  \big)
  \dcableForce{i},
\label{eq:qdot_from_fdot}
\end{equation}
the force derivative has a nonzero component orthogonal to the cable
direction, proving the claim.
\end{proof}
\begin{figure}[htbp]
  \centering
  \begin{minipage} {0.48\columnwidth}
{\includegraphics[width=\linewidth]{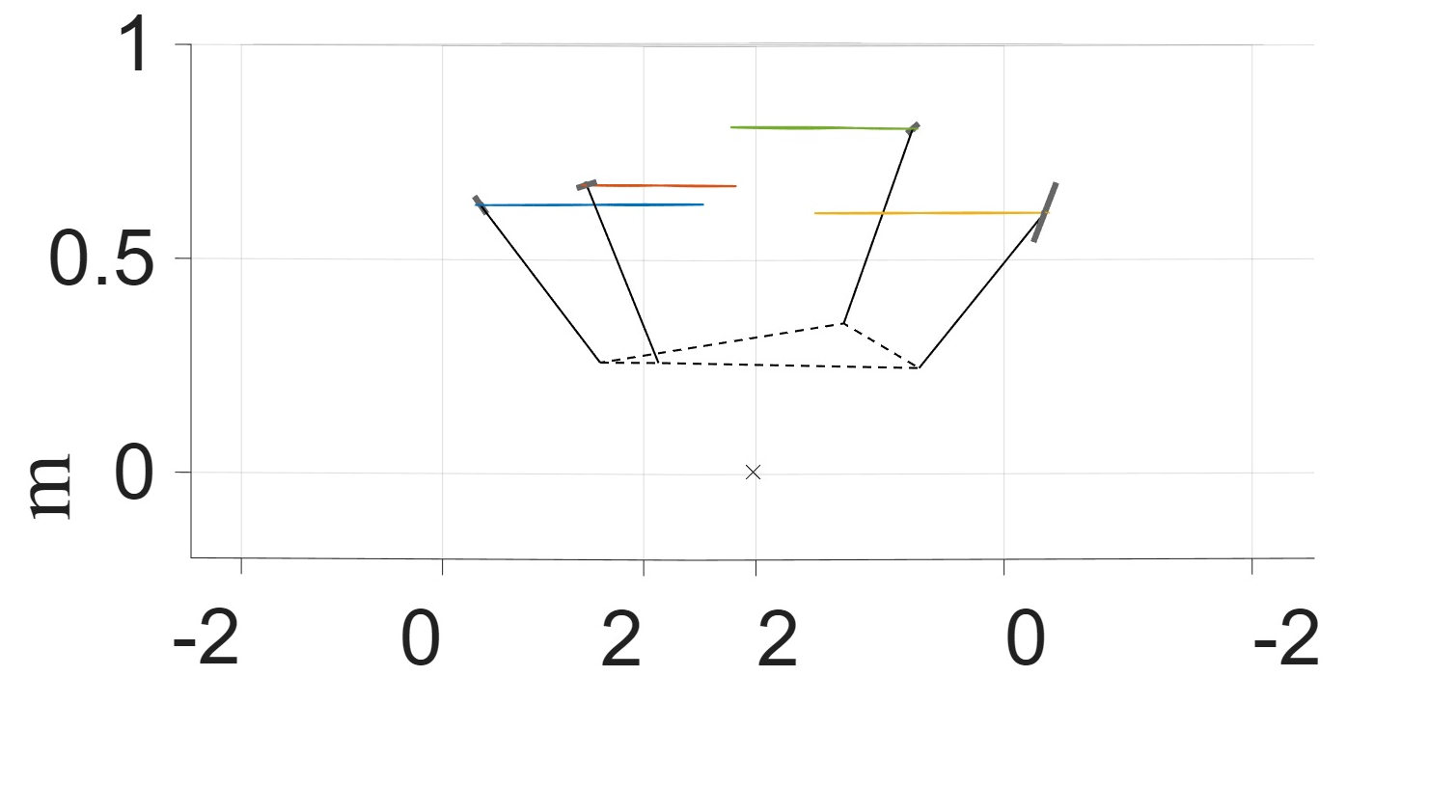}}
{\includegraphics[width=\linewidth]{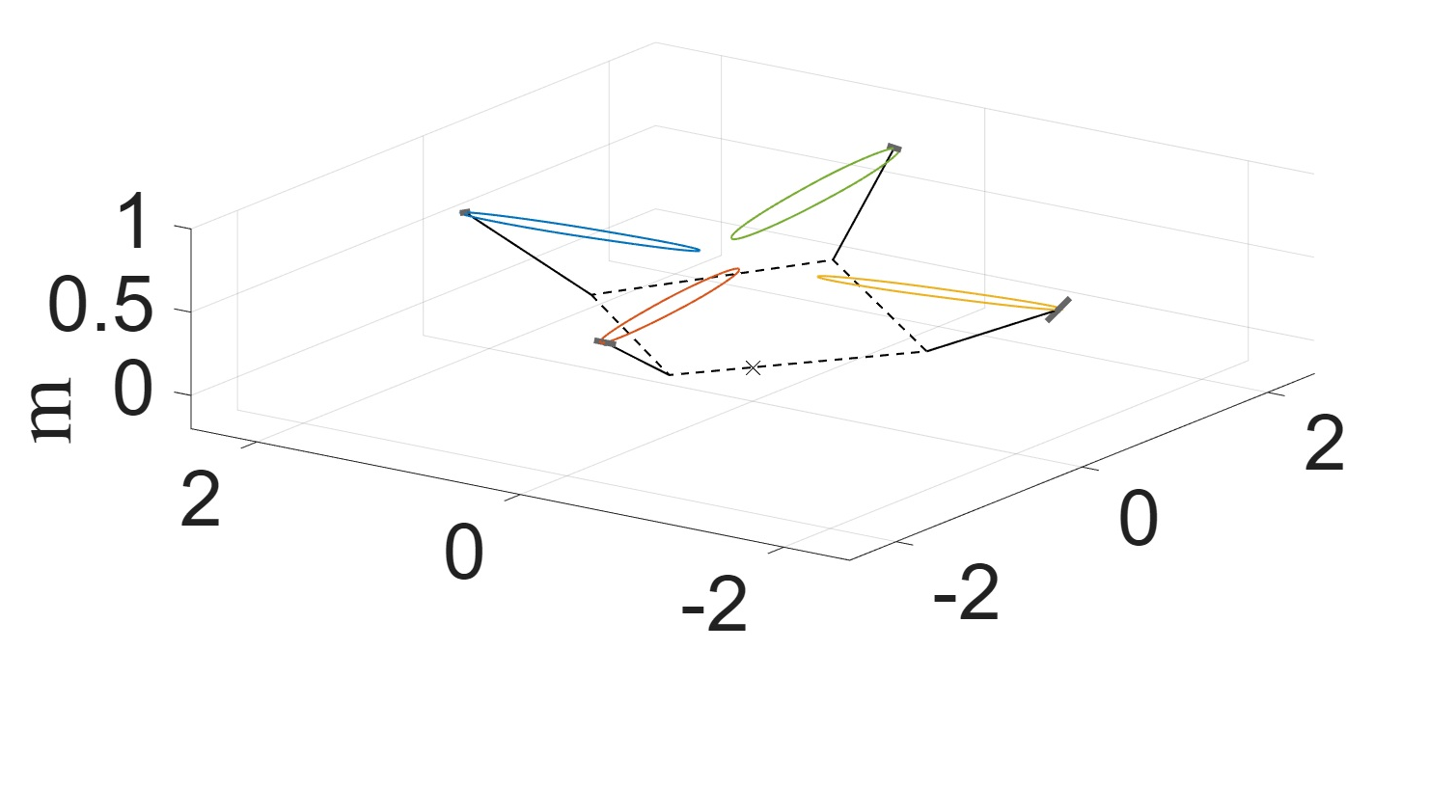}}
\caption*{$\ell_i$ as in \eqref{eq:level_length_construction}}
  \end{minipage}
  \begin{minipage}{0.48\columnwidth}
\includegraphics[width=\linewidth]{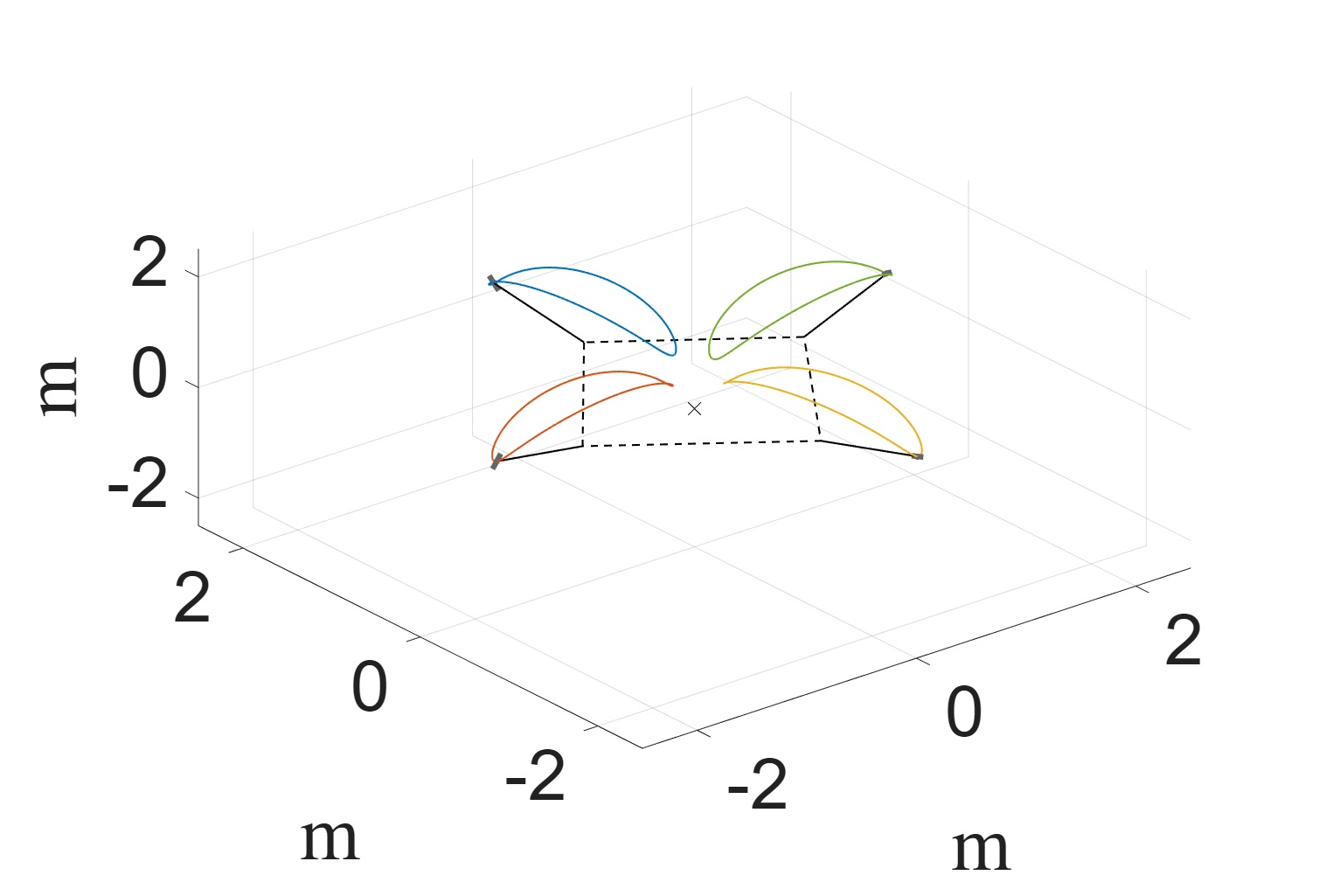}
\caption*{$\dot{\ell}_i=0$}
  \end{minipage}
  \begin{minipage} {0.48\columnwidth}
{\includegraphics[width=\linewidth]{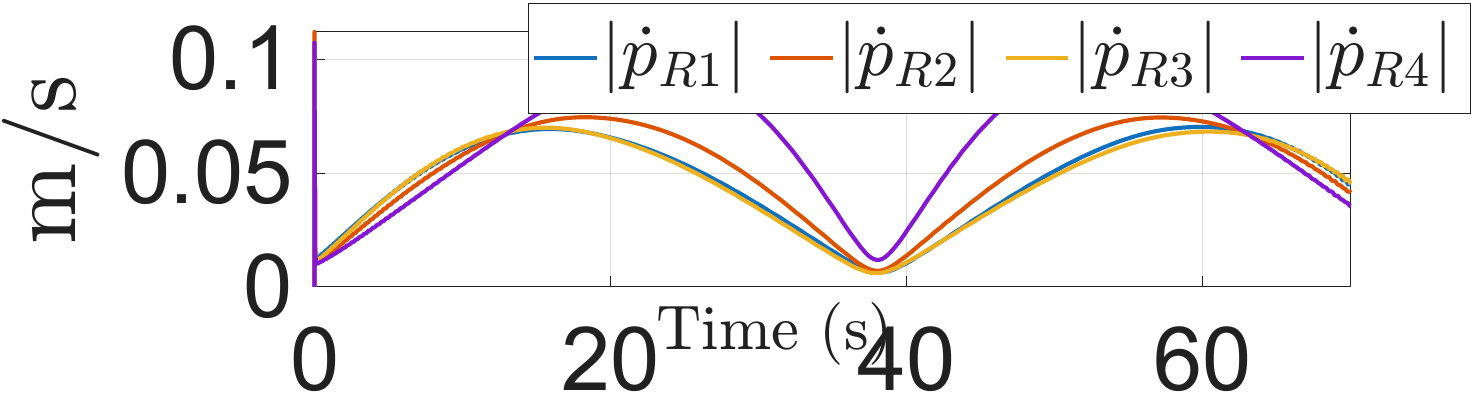}}
\caption*{$\ell_i$ as in \eqref{eq:level_length_construction}}
\end{minipage}
\begin{minipage}{0.48\columnwidth}
{\includegraphics[width=\linewidth]{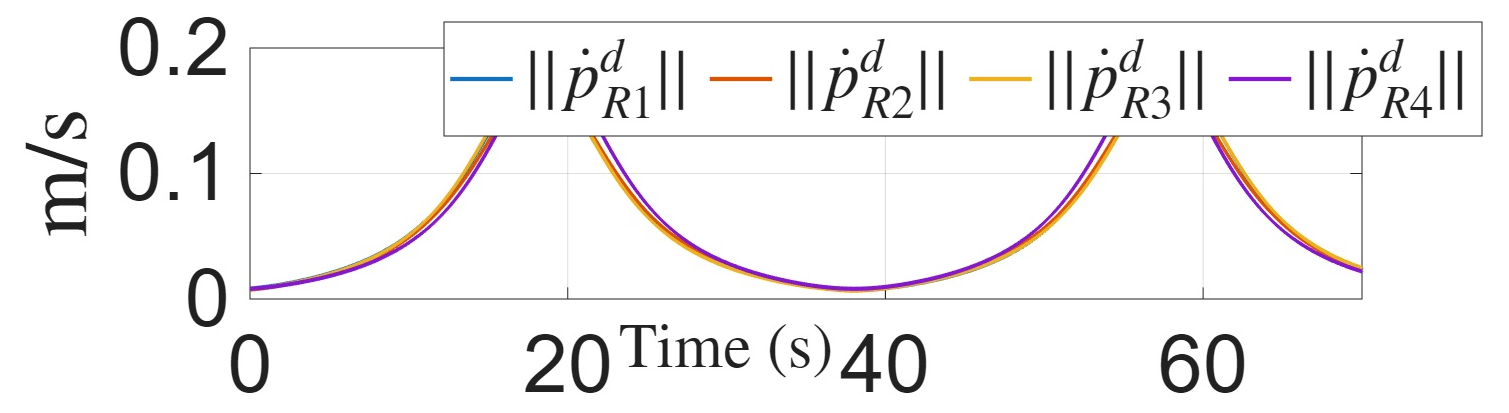}}
\caption*{$\dot{\ell}_i=0$}\end{minipage}
  \subfloat{\includegraphics[width=0.75\linewidth]{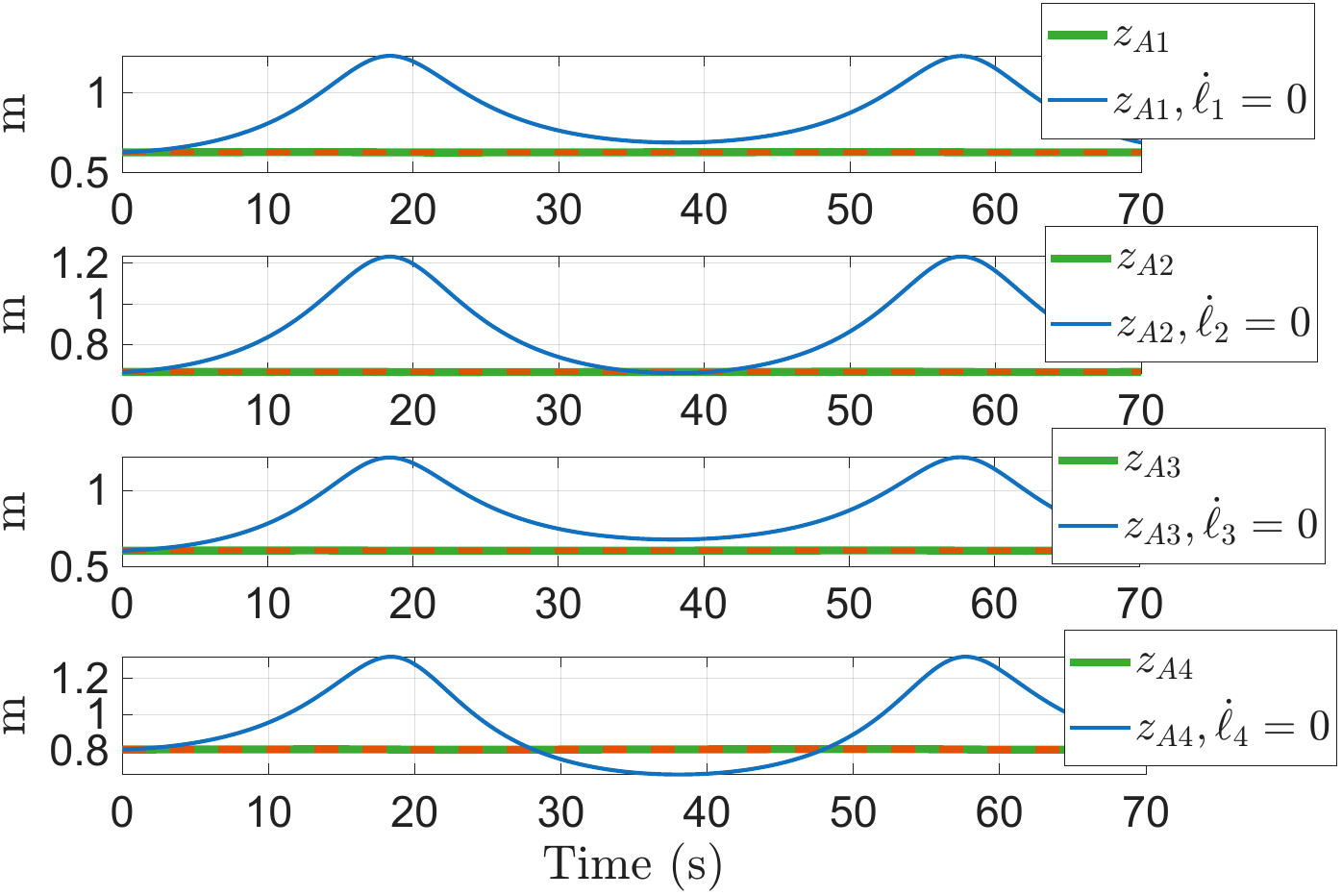}}
  \hfill
\subfloat{\includegraphics[width=0.8\linewidth]{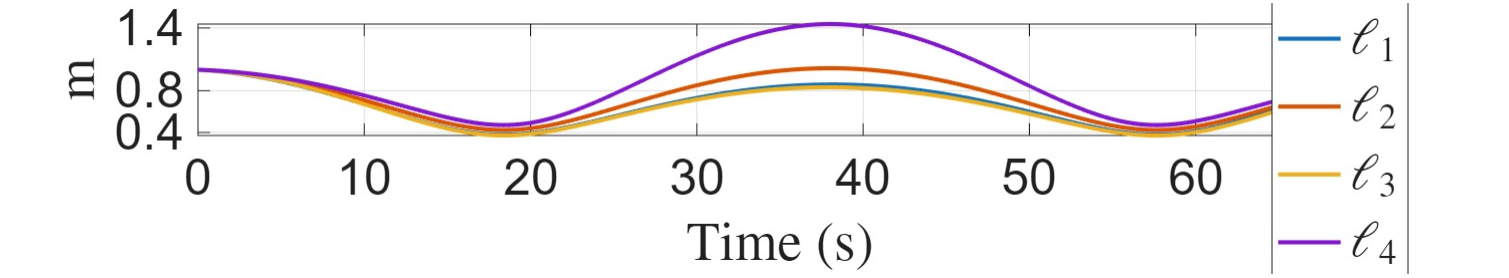}}
  \hfill
  \caption{Results with cable length variation (for level flight) and without ($\dot{\ell}_i=0$). In the first row, a schematic representation of the systems, with the UAV paths highlighted with colored lines. 
The carrier speeds remain positive in both cases (second row). The third row shows that oscillations in the UAV altitudes for $\dot{\ell}_i=0$ are of several decimetres, while their amplitude is reduced to around $0.001\si{\metre}$  with cable length variations, reported in the last row.}
  \label{fig:main_label}
\end{figure}
\begin{theorem}[Level equilibrium-preserving non-stop motion]
\label{thm:level_nonstop_motion}
Let \(n\geq3\), and suppose that, for every \(i=1,\ldots,n\),
\begin{equation}
  \rank
  \begin{bmatrix}
    \coord{w}{\delta_i}
    &
    \coord{w}{\bar{\delta}_i}
  \end{bmatrix}
  =2,
  \qquad
  \particularForceBlock{i}
  \notin
  \internalPlane{i}.
\label{eq:assumption_n3}
\end{equation}
Choose the cycle parameters as
\begin{equation}
  \lambda_j(t)
  =
  \lambda_{j0}
  +
  a\sin(\xi t+\varphi_j),
  \qquad
  j=1,\ldots,n,
\label{eq:phase_shifted_cycle_parameters}
\end{equation}
where \(a\neq0\), \(\xi>0\), and the phases satisfy
\begin{equation}
  \varphi_{h_i}
  -
  \varphi_{h_i^+}
  \notin
  \pi\mathbb Z,
  \qquad
  i=1,\ldots,n.
\label{eq:adjacent_phase_condition}
\end{equation}
Assume that the resulting force trajectory remains in
\(\mathcal F_+\), that
$  \abs{q_{i,z}(t)}
  \geq
  \varepsilon_i
  >0,
  \qquad
  \forall t\geq0,
  i=1,\ldots,n,
$ and that
\begin{equation}
  \cableLength{i}(t)
  =
  \frac{
    h_i-z_{B_i}
  }{
    q_{i,z}(t)
  },
  \qquad
  i=1,\ldots,n,
\label{eq:level_length_construction}
\end{equation}
satisfy their admissible bounds. Then, the resulting carrier
trajectories preserve the load equilibrium, remain at the constant
altitudes \(h_i\), and are non-stop.
\end{theorem}
\ifarXiv
\begin{figure*}
    \centering
     \begin{minipage}{\textwidth}
    \includegraphics[width=0.33\linewidth]{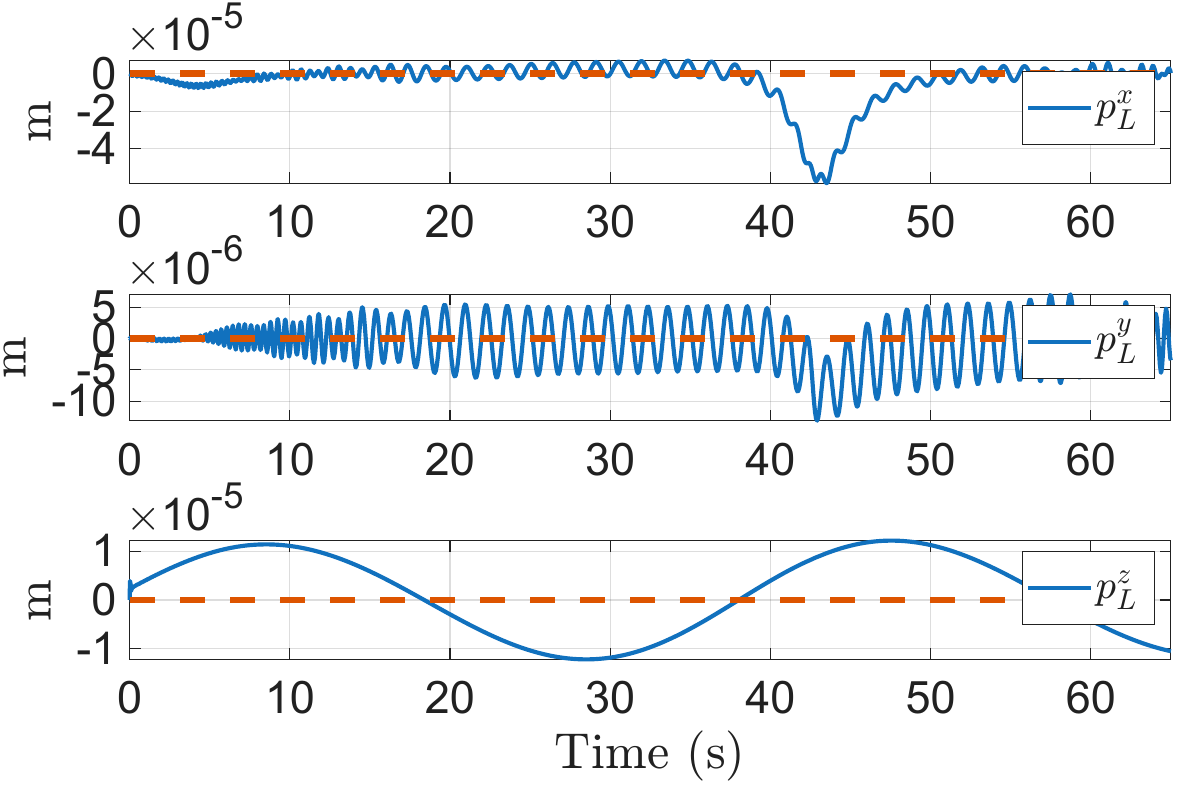}
    \includegraphics[width=0.33\linewidth]{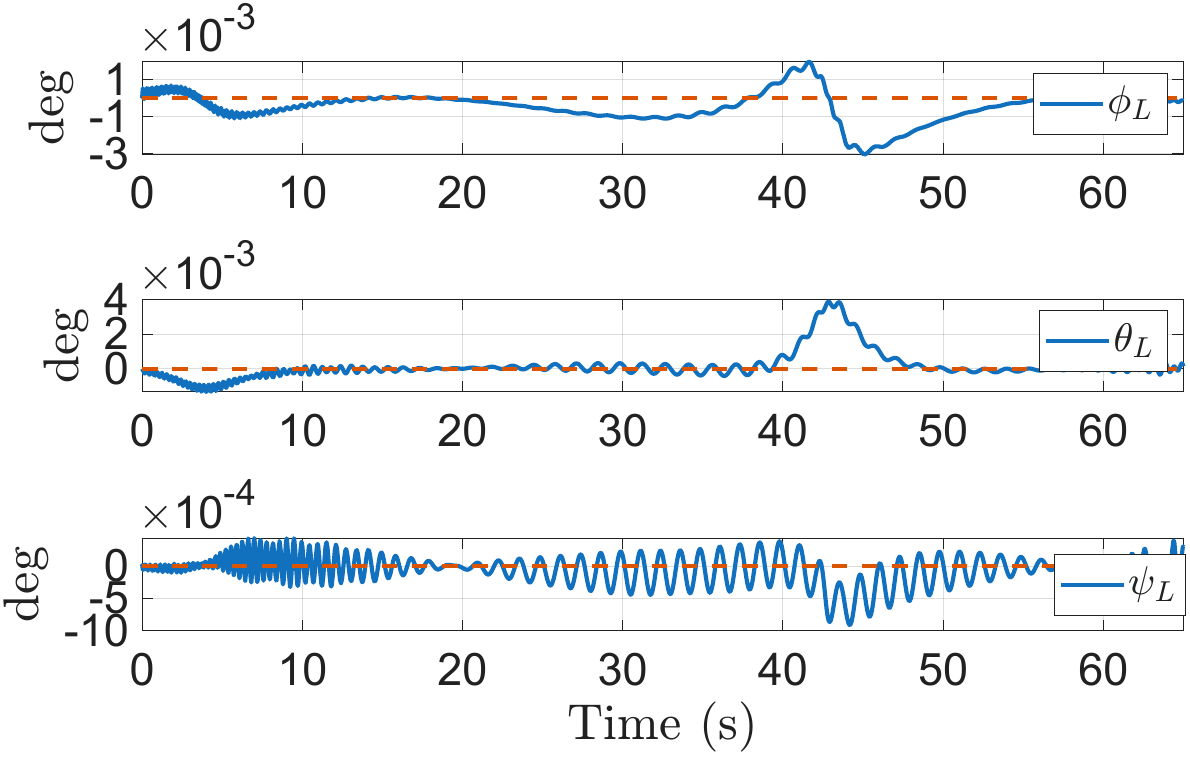}
           \includegraphics[width=0.30\linewidth]{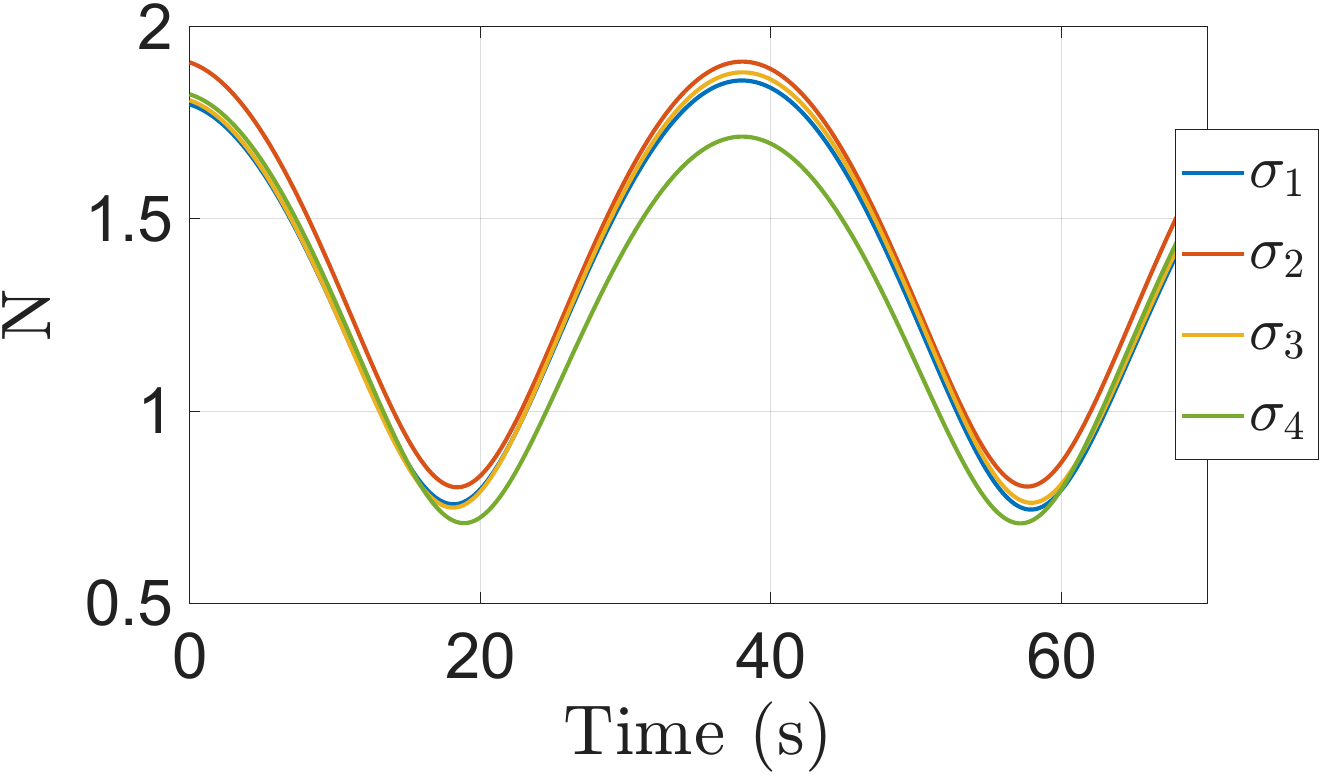}
          \caption*{$\ell_i$ as in \eqref{eq:level_length_construction}}
\end{minipage}
\begin{minipage}{\textwidth}
     \includegraphics[width=0.33\linewidth]{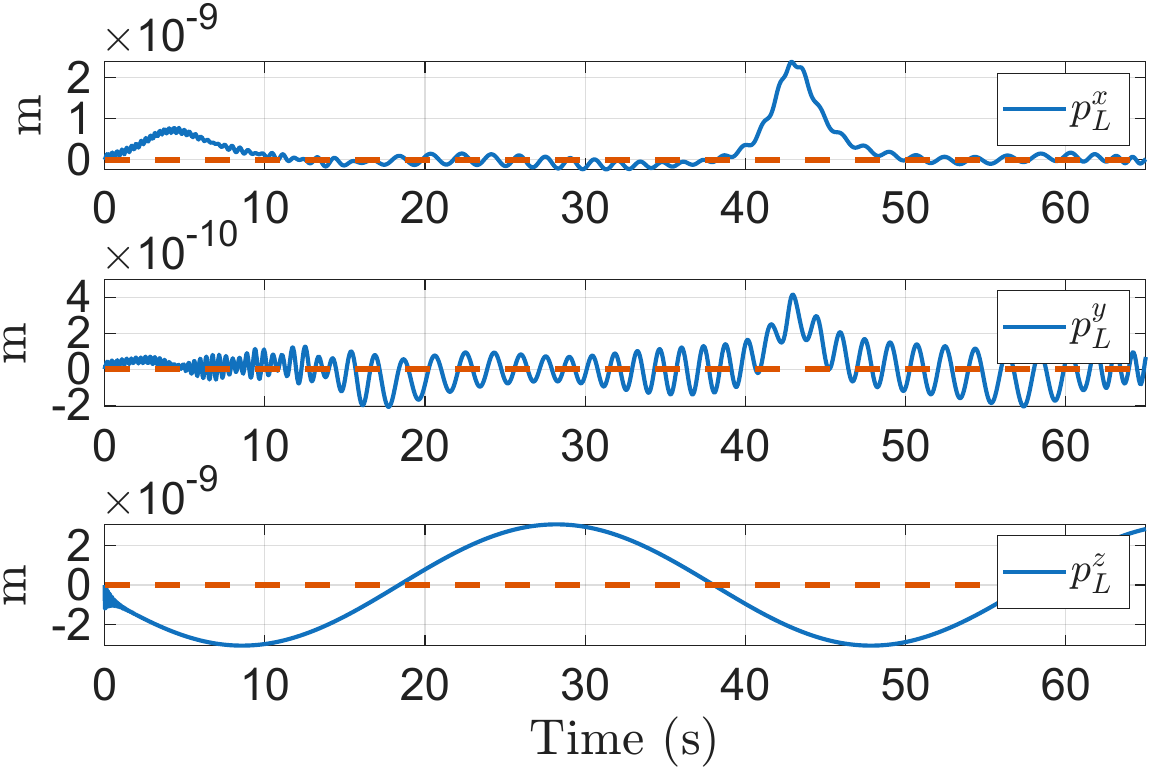}
       \includegraphics[width=0.33\linewidth]{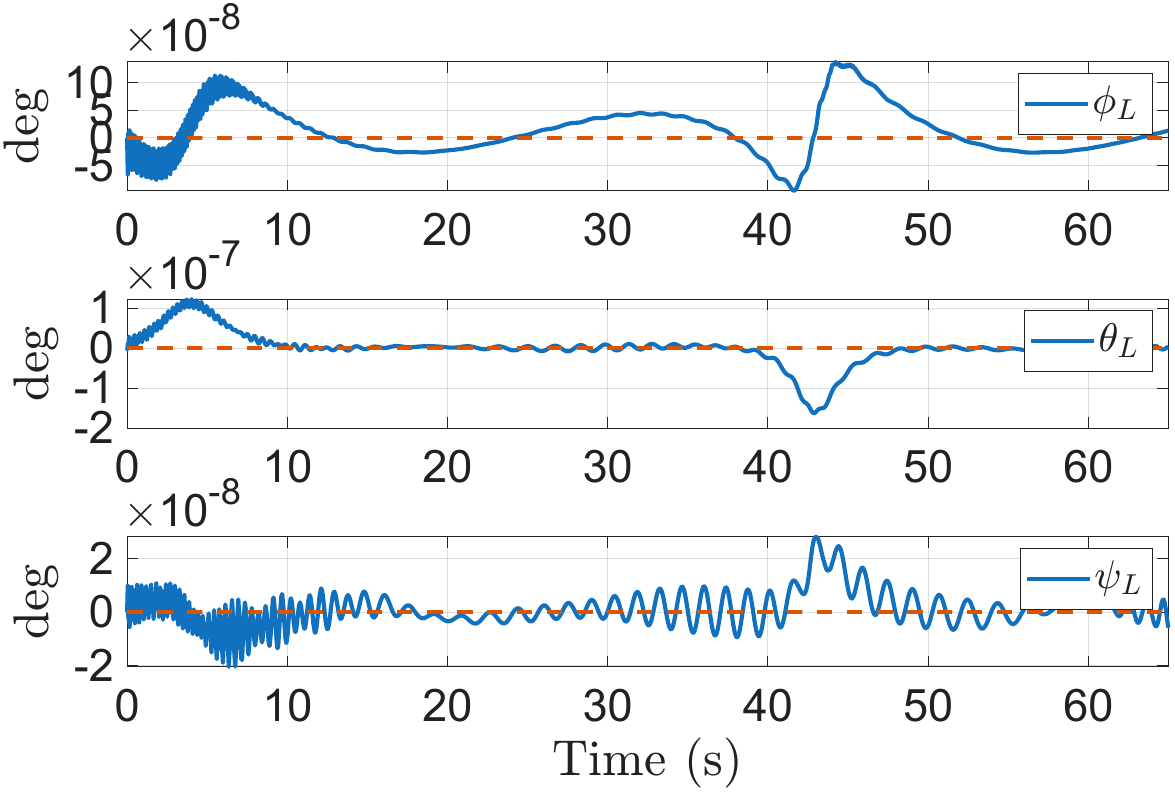}
              \includegraphics[width=0.30\linewidth]{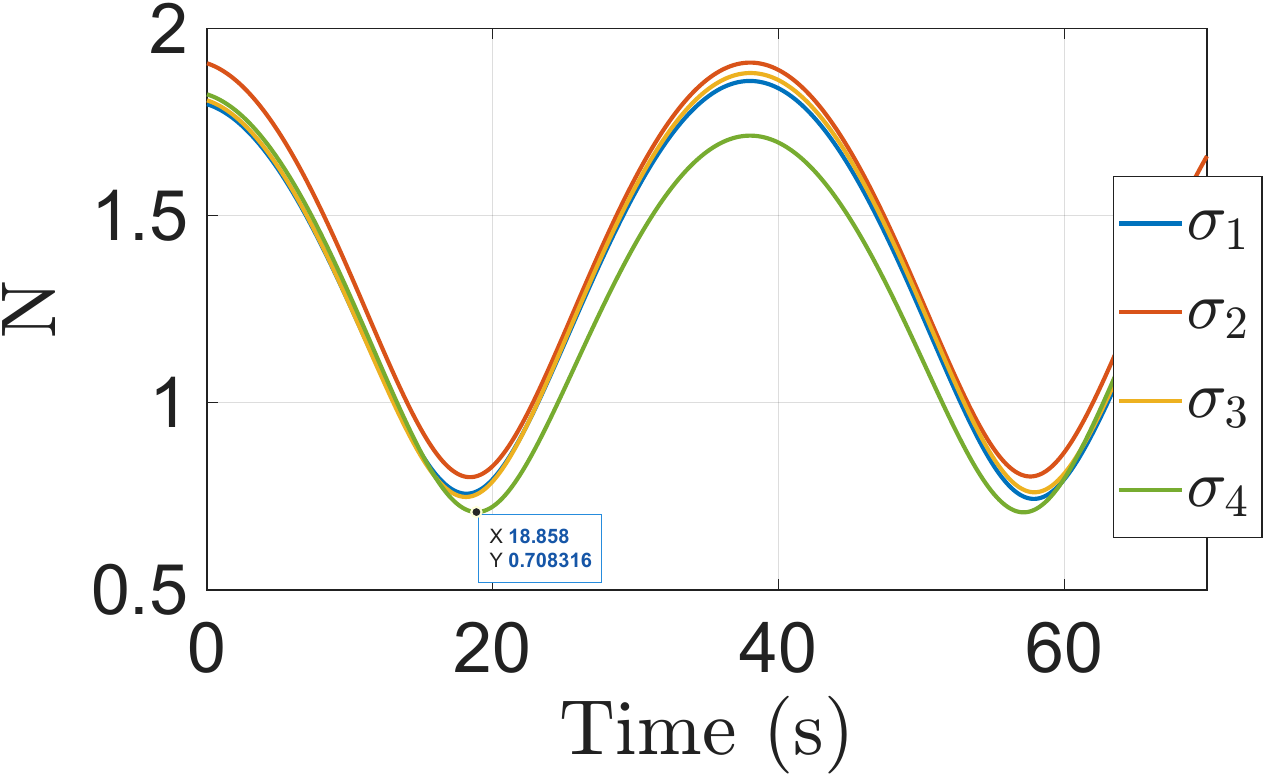}
\caption*{$\dot{\ell}_i=0$}
\end{minipage}
    \caption{Load pose and cable tension with length variation for level flight (first row) and without ($\dot{\ell}_i=0$). }
    \label{fig:pL-force}
\end{figure*}
\fi
\begin{proof}
Differentiating \eqref{eq:phase_shifted_cycle_parameters} and using
\eqref{eq:fdot_plane} gives
\begin{equation}
  \dcableForce{i}
  =
  a\xi
  \left[
    \cos(\xi t+\varphi_{h_i})
    \coord{w}{\delta_i}
    +
    \cos(\xi t+\varphi_{h_i^+})
    \coord{w}{\bar{\delta}_i}
  \right].
\label{eq:phase_shifted_force_derivative}
\end{equation}
By \eqref{eq:assumption_n3}, the two local direction arrays are
linearly independent. Therefore, \(\dcableForce{i}=\bm0\) could occur
only if
\[
  \cos(\xi t+\varphi_{h_i})=0,
  \qquad
  \cos(\xi t+\varphi_{h_i^+})=0.
\]
These equalities can hold simultaneously only if
\(
  \varphi_{h_i}-\varphi_{h_i^+}\in\pi\mathbb Z
\),
which is excluded by \eqref{eq:adjacent_phase_condition}. Hence,
\(\dcableForce{i}(t)\neq\bm0\) for every \(i\) and \(t\).
Lemma~\ref{lemma:force_direction_motion} then gives
\(
  \dcableDir{i}(t)\neq\bm0
\)
for every \(i\) and \(t\geq0\). By Corollary~\ref{cor:constant_altitude_realization},
\eqref{eq:level_length_construction} realizes
\(z_{A_i}(t)=h_i\).
Proposition~\ref{prop:constant_altitude_stopping_condition} therefore
gives
\(
  \dcarrierPos{i}(t)\neq\bm0
\)
for every \(i\) and \(t\geq0\).
Continuity and periodicity then yield a common positive minimum
\(\speedLowerBound>0\) over one period and all carriers.
\end{proof}

{A feasible phase selection always exists. For example, indexing
the cycle edges in their traversal order, one may choose
\(
  \varphi_j=(j-1)\pi/n
\).
The phase difference between every pair of consecutive cycle edges,
including the last and first edges, is then not an integer multiple of
\(\pi\).}

\begin{corollary}[Minimum number for level-flight construction]
\label{cor:minimum_level_carriers}
{Within the periodic internal-force constructions considered
above, unrestricted admissible non-stop motion requires at least two
carriers, whereas admissible level non-stop motion requires at least
three.}
\end{corollary}

\begin{proof}
The unrestricted claim follows from
Corollary~\ref{cor:minimum_unrestricted_carriers}.
Proposition~\ref{prop:two_carrier_level_obstruction} excludes periodic
level non-stop motion for two carriers, while
Theorem~\ref{thm:level_nonstop_motion} establishes its existence for
\(n\geq3\) under the stated geometric and feasibility assumptions.
\end{proof}
\section{Numerical Illustration}
\label{sec:numerical_validation}

Numerical simulations in Figure \ref{fig:main_label} illustrate the level non-stop construction of
Theorem~\ref{thm:level_nonstop_motion}. A rigid load of mass equal to $0.3\si{\kilogram}$ and diagonal rotational inertia with diagonal elements equal to $0.01\si{\kilogram} \cdot \si{\metre}^2$ is manipulated
by \(n=4\) aerial carriers. Cable elasticity is introduced as a non-ideality and is beneficial for the simulator's modularity: massless elastic cables with have a stiffness equal to $10^4\si{\newton\per\metre}$, a damping coefficient equal to $10^3\si{\newton\second\per\metre}$, and an initial rest length equal to $1\si{\metre}$. The load-side end points of the cables, generated as  $\attachmentCoord{i}=[R_z(\frac{2\pi}{in}+0.2r[1.2/, 0]/, r]^\top$ with $r$ a random number in $(0, 1)$, are $ \attachmentCoord{1}=[-0.036\, 1.200\, 0.256]^\top\si{\metre}$, $\attachmentCoord{2}=[-1.183\, -0.201\, 0.254]^\top\si{\metre}$, $\attachmentCoord{3}=[0.195\, -1.184\, 0.244]^\top\si{\metre}$, and   $\attachmentCoord{3}=[1.180\, 0.222\, 0.350]^\top\si{\metre}$.  The
carrier-side endpoints exactly realize the trajectories generated by
\eqref{eq:phase_shifted_cycle_parameters} with $ \lambda_{j0}=0$, $a=1.2\si{newton}$
  $\xi= 0.08\si{\radian\per\second}$, $\varphi_1=\varphi_3=1.77\si{\radian}$, and $\varphi_2=\varphi_4=\pi/2\si{\radian}$; the load motion results
from the simulated rigid-body dynamics. The cable length follows  
\eqref{eq:level_length_construction} with $h_i$ corresponding to the initial altitudes, respectively: $h_1=0.627\si{\metre}$, $h_2=0.670\si{\metre}$, $h_3=0.606\si{\metre}$, and $h_4=0.807\si{\metre}$.
\ifarXiv
Figure \ref{fig:pL-force} shows that the load pose was unperturbed, and the minimum cable tension was positive and equal to $0.71\si{\newton}$ in both cases.
\fi

\section{Conclusions and Limitations}
\label{sec:conclusions}

Variable cable lengths provide radial freedom for shaping
equilibrium-preserving carrier trajectories and realizing prescribed
directional-position profiles, including constant altitude. While two
carriers suffice for unrestricted non-stop motion, periodic level
motion is obstructed for two carriers and can be constructed for
\(n\geq3\) under the stated feasibility conditions. The analysis assumes
taut cables, ideal cable-length and carrier-trajectory tracking, and a constant load pose,
while neglecting winch and carrier constraints, cable slackness,
aerodynamics, and collisions. Future work will address that through constrained trajectory generation, winch dynamics, and feedback control of
time-varying load poses.

\bibliographystyle{ieeetran}
\bibliography{biblio}

\end{document}

%% file: symbols.tex
\definecolor{afmodcolor}{RGB}{0,105,125}

\newcommand{\Eaff}{\mathbb{E}}                  
\newcommand{\Vtrans}{\mathbb{V}}               
\newcommand{\WrenchSpace}{\mathbb{W}}          

\newcommand{\Rthree}{\mathbb{R}^{3}}

\newcommand{\SOthree}{\mathrm{SO}(3)}
\newcommand{\SEthree}{\mathrm{SE}(3)}

\newcommand{\sphere}[1]{\mathbb{S}^{#1}}

\newcommand{\Frame}[1]{\mathcal{F}_{\mathsf{#1}}}
\newcommand{\coord}[2]{\bm{#2}^{\mathsf{#1}}}
\newcommand{\pointcoord}[2]{\bm{p}_{#2}^{\mathsf{#1}}}
\newcommand{\rot}[2]{\bm{R}_{\mathsf{#2}}^{\mathsf{#1}}}
\newcommand{\pose}[2]{\bm{T}_{\mathsf{#2}}^{\mathsf{#1}}}

\newcommand{\dcoord}[2]{\dot{\bm{#2}}^{\mathsf{#1}}}

\newcommand{\dpointcoord}[2]{\dot{\bm{p}}_{#2}^{\mathsf{#1}}}
\newcommand{\ddpointcoord}[2]{\ddot{\bm{p}}_{#2}^{\mathsf{#1}}}

\newcommand{\define}{:=}
\newcommand{\norm}[1]{\left\lVert #1\right\rVert}
\newcommand{\abs}[1]{\left\lvert #1\right\rvert}

\newcommand{\inner}[2]{\left\langle #1,#2\right\rangle}
\newcommand{\col}[1]{\operatorname{col}\!\left(#1\right)}

\newcommand{\im}{\operatorname{im}}
\newcommand{\rank}{\operatorname{rank}}

\newcommand{\ith}[1]{#1\text{-th}}

\newcommand{\fw}{w}                              
\newcommand{\fb}{b}                              

\newcommand{\frameW}{\Frame{\fw}}
\newcommand{\frameB}{\Frame{\fb}}

\newcommand{\originW}{O_{\mathsf{\fw}}}
\newcommand{\originB}{O_{\mathsf{\fb}}}

\newcommand{\anchorPoint}[1]{B_{#1}}             
\newcommand{\carrierPoint}[1]{A_{#1}}            

\newcommand{\loadPos}{\pointcoord{\fw}{\originB}}
\newcommand{\dloadPos}{\dpointcoord{\fw}{\originB}}
\newcommand{\ddloadPos}{\ddpointcoord{\fw}{\originB}}
\newcommand{\loadRot}{\rot{\fw}{\fb}}
\newcommand{\loadPose}{\pose{\fw}{\fb}}

\newcommand{\loadMass}{m_{\mathrm L}}
\newcommand{\loadInertia}{\bm J_{\mathrm L}^{\mathsf b}}
\newcommand{\gravityScalar}{g}
\newcommand{\worldVertical}{%
  \coord{\fw}{z_{\mathsf w}}%
}

\newcommand{\loadAngVel}{\coord{\fb}{\omega_{\mathrm L}}}
\newcommand{\dloadAngVel}{\dcoord{\fb}{\omega_{\mathrm L}}}

\newcommand{\cableLength}[1]{\ell_{#1}}
\newcommand{\dcableLength}[1]{\dot{\ell}_{#1}}

\newcommand{\cableLengthMin}{\underline{\ell}}
\newcommand{\cableLengthMax}{\overline{\ell}}
\newcommand{\cableTension}[1]{\sigma_{#1}}
\newcommand{\dcableTension}[1]{\dot{\sigma}_{#1}}

\newcommand{\cableDirIntrinsic}[1]{q_{#1}}

\newcommand{\attachmentCoord}[1]{\coord{\fb}{b_{#1}}}
\newcommand{\attachmentWorld}[1]{\coord{\fw}{r_{#1}}}
\newcommand{\cableDir}[1]{\coord{\fw}{q_{#1}}}
\newcommand{\dcableDir}[1]{\dcoord{\fw}{q_{#1}}}
\newcommand{\cableForce}[1]{\coord{\fw}{f_{#1}}}
\newcommand{\dcableForce}[1]{\dcoord{\fw}{f_{#1}}}
\newcommand{\carrierPos}[1]{\pointcoord{\fw}{A_{#1}}}
\newcommand{\dcarrierPos}[1]{\dpointcoord{\fw}{A_{#1}}}

\newcommand{\cableForceStack}{\coord{\fw}{f}}

\newcommand{\internalParam}{\bm\eta}
\newcommand{\cycleParam}{\bm\lambda}

\newcommand{\Grasp}{\mathcal{G}_{\originB}}
\newcommand{\grasp}{\bm G_{\originB}^{\mathsf w}}
\newcommand{\equilibriumWrench}{\bm w_{\mathrm{eq}}^{\mathsf w}}

\newcommand{\nullBasis}{\bm N_{G}^{\mathsf w}}
\newcommand{\cycleNullMatrix}{\bm N_{H}^{\mathsf w}}
\newcommand{\particularForce}{\bm f_{0}^{\mathsf w}}
\newcommand{\particularForceBlock}[1]{\bm f_{0#1}^{\mathsf w}}
\newcommand{\internalForceStack}{\widetilde{\bm f}^{\mathsf w}}
\newcommand{\internalForceBlock}[1]{\widetilde{\bm f}_{#1}^{\mathsf w}}

\newcommand{\hamCycle}{\mathcal{H}}

\newcommand{\incidence}{\bm H}
\newcommand{\edgeDir}[2]{\coord{\fw}{d_{#1#2}}}
\newcommand{\internalPlane}[1]{\mathcal{P}_{#1}}

\newcommand{\speedLowerBound}{\underline{v}}

\newtheorem{definition}{Definition} 
 
\newtheorem{prop}{Proposition} 
\newtheorem{corollary}{Corollary} 
\newtheorem{lemma}{Lemma}

\newtheorem{theorem}{Theorem}